\documentclass{article}

\usepackage{color}
\usepackage{graphicx} 
\usepackage{mathtools}
\usepackage{amsmath,amsfonts,amssymb,amsthm}
\usepackage{ascmac}
\usepackage{lineno}
\usepackage{fancybox}
\usepackage{authblk}
\usepackage{makecell}
\usepackage[letterpaper,top=2cm,bottom=2cm,left=3cm,right=3cm,marginparwidth=1.75cm]{geometry}
\usepackage{todonotes}
\usepackage{url}
\usepackage{hyperref}
\usepackage{cite}
\usepackage{cleveref}
\usepackage{tabularx}
\usepackage{multirow} 
\usepackage{booktabs} 
\usepackage{comment}

\newtheorem{theorem}{Theorem}
\newtheorem{proposition}{Proposition}
\newtheorem{corollary}{Corollary}

\newtheorem{observation}{Observation}

\newtheorem{claim}{Claim}

\newcommand{\defmodel}[5]{
    \vspace{3mm}
    \noindent\fbox{
        \begin{minipage}{.95\columnwidth}
            #1\newline
            \textbf{Input:} #2\\
            \textbf{Strategy profile:} #3\\
            \textbf{Cost of an agent:} #4\\
            \textbf{Total cost:} #5
        \end{minipage}
    }
    \vspace{3mm}
}

\newenvironment{proofclaim}{\noindent{\em Proof of the claim.}}{\qedclaim}
\newcommand{\qedclaim}{\hfill $\diamond$ \medskip}

\newcommand{\vc}{{\rm vc}}

\newcommand{\ccMAR}{\textsc{CC-MAR}}
\newcommand{\PS}{\mathcal{P}}
\newcommand{\cost}{\mathrm{cost}}

\newcommand{\terminals}{{\mathcal{T}}}

\newcommand{\tesshu}[1]{\textcolor{purple}{{\textbf{Te: }#1}}}

\title{Game-Theoretic and Algorithmic Analyses of Multi-Agent Routing under Crossing Costs\thanks{Tesshu Hanaka is partially supported by JSPS KAKENHI Grant Numbers JP21K17707, JP22H00513, JP23H04388, JP25K03077, and JST, CRONOS, Japan Grant Number JPMJCS24K2. Nikolaos Melissinos is partially supported by Charles University projects UNCE 24/SCI/008 and PRIMUS 24/SCI/012, and by the project 25-17221S of GAČR. Hirotaka Ono is partially supported by JSPS KAKENHI Grant Numbers JP20H05967, JP22H00513, JP24K02898, JP25K03077, and JST, CRONOS, Japan Grant Number JPMJCS24K2.}}

\author[1]{Tesshu Hanaka}
\author[2]{Nikolaos Melissinos}
\author[3]{Hirotaka Ono}

\affil[1]{Kyushu University, Fukuoka, Japan\\ \texttt{hanaka@inf.kyushu-u.ac.jp}}
\affil[2]{Computer Science Institute, Faculty of Mathematics and Physics, Charles University, Prague, Czech Republic\\ \texttt{melissinos@iuuk.mff.cuni.cz}} 
\affil[3]{Nagoya University, Nagoya, Japan\\ \texttt{ono@nagoya-u.jp}}

\date{}

\begin{document}

\maketitle

\begin{abstract}
Coordinating the movement of multiple autonomous agents over a shared network is a fundamental challenge in algorithmic robotics, intelligent transportation, and distributed systems. The dominant approach, Multi-Agent Path Finding (MAPF), relies on centralized control and synchronous collision avoidance, which often requires strict synchronization and guarantees of globally conflict-free execution.
This paper introduces the Multi-Agent Routing under Crossing Cost model on mixed graphs, a novel framework tailored for asynchronous settings. In our model, instead of treating conflicts as hard constraints, each agent is assigned a path, and the system is evaluated through a cost function that measures potential head-on encounters. This ``crossing cost'', which is defined as the product of agents traversing an edge in opposite directions, quantifies the risk of congestion and delay in decentralized execution.

Our contributions are both game-theoretic and algorithmic. We model the setting as a congestion game with a non-standard cost function, prove the existence of pure Nash equilibria, and analyze the dynamics leading to them. Equilibria can be found in polynomial time under mild conditions, while the general case is PLS-complete.
From an optimization perspective, deciding whether a solution with zero crossing cost exists generalizes the Steiner Orientation problem, making the problem NP-complete. To address this hardness barrier, we design a suite of parameterized algorithms for minimizing crossing cost, with parameters including the number of arcs, edges,  agents, and structural graph measures such as vertex cover. These yield XP or FPT results depending on the parameter, offering algorithmic strategies for structurally restricted instances. 
Our framework provides a new theoretical foundation for decentralized multi-agent routing, bridging equilibrium analysis and parameterized complexity to support scalable and risk-aware coordination.
\end{abstract}

\section{Introduction}

Coordinating the movement of multiple autonomous agents over a shared network is a fundamental challenge in algorithmic robotics, intelligent transportation, and distributed systems. 
A central abstraction for this task is the Multi-Agent Path Finding (MAPF) problem, which has numerous real-world applications, e.g., in warehouse management, autonomous vehicle coordination, and multi-robot systems~\cite{DBLP:journals/aim/WurmanDM08,DBLP:conf/aaai/0001TKDKK21,DBLP:conf/aaai/Yan025}. In MAPF, the goal is to compute collision-free paths that prevent agents from simultaneously occupying the same vertex or traversing the same edge in opposite directions. The canonical MAPF formulation typically assumes centralized control, synchronized movement, and tight coordination to enforce these constraints, with the primary objective being the minimization of makespan or total cost.

In contrast, many real-world multi-agent systems, such as fleets of delivery robots, autonomous vehicles on road networks, or agents in communication-limited environments, operate with asynchronous timing and decentralized decision-making. In such settings, it is often impractical to enforce strict synchronization or guarantee globally conflict-free execution. Instead, agents may be provided with routes (that is, sequences of vertices to be followed), while the precise timing of their execution remains unspecified. Consequently, potential conflicts are better viewed not as hard constraints but as quantifiable risks arising from route interactions.

A powerful and well-studied framework for modeling such scenarios is the class of network congestion games, introduced by Rosenthal~\cite{Rosenthal1973}. In these games, each agent has a pair of vertices: a source and a destination. Each agent then selfishly selects a path between its source and destination on a network to minimize a cost function based on edge congestion, typically defined as a non-decreasing function of the number of agents traversing an edge. This formulation reflects situations where agents operate independently and may interact only partially or unpredictably during execution.
It is well-known that network congestion games are guaranteed to possess a pure-strategy Nash equilibrium~\cite{Rosenthal1973}. While finding such an equilibrium is PLS-complete in general~\cite{stoc/FabrikantPT04,AckermannRV08}, it is tractable in specific cases, such as when all agents share the same source-terminal pair~\cite{stoc/FabrikantPT04}.

In this work, we introduce a novel congestion game model, which we call \textbf{Crossing Cost Multi-Agent Routing (CC-MAR)}. Unlike traditional network congestion games, the cost is incurred not by the total number of agents on an edge, but only by those traversing it in opposite directions; only potential head-on encounters are penalized. This asymmetric cost structure reflects practical scenarios such as unidirectional flow in narrow corridors or lanes, where same-direction usage is typically non-disruptive.
Figure~\ref{fig:example} shows an example of our model.

\begin{figure}[h]
    \centering
    \includegraphics[width=0.5\linewidth]{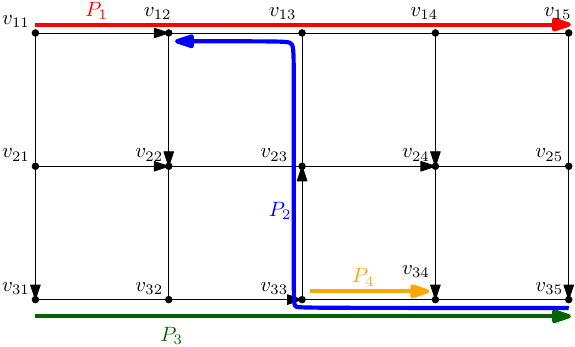}
    \caption{An example of CC-MAR. In this example, there are 4 agents with terminal pairs $(v_{11},v_{15})$, $(v_{35},v_{12})$, $(v_{31},v_{35})$, and $(v_{33},v_{34})$, respectively. 
The strategy of each agent is $P_i$ for $i\in \{1,2,3,4\}$. Path $P_1$ and path $P_2$ cross each other on the edge $\{v_{12},v_{13}\}$. Thus, the crossing cost between $P_1$ and $P_2$ is 1. The individual cost of agent 1 is also 1, as it has no other crossings.
On the other hand, the individual cost of agent 2 is 4. This is because it crosses $P_1$ on edge $\{v_{12},v_{13}\}$, crosses $P_3$ on edges $\{v_{33},v_{34}\},\{v_{34},v_{35}\}$, and crosses $P_4$ on edge $\{v_{34},v_{35}\}$. The total cost of the strategy profile is the sum of all pairwise crossing costs, which is $1 (\text{from } P_1, P_2) + 2 (\text{from } P_2, P_3) + 1 (\text{from } P_2, P_4) = 4$.}
    \label{fig:example}
\end{figure}

\begin{table*}[!t]
\centering
\caption{Summary of our contributions for the Crossing Cost Multi-Agent Routing (CC-MAR) problem.}
\label{tab:contributions}
\renewcommand{\arraystretch}{1.2}
\begin{tabularx}{\linewidth}{@{}ll >{\raggedright\arraybackslash}X@{}} 
\toprule
\textbf{Perspective} & \textbf{Problem Setting} & \textbf{Main Contribution \& Complexity} \\
\midrule
\multirow{2}{*}{\begin{tabular}[c]{@{}l@{}}\textbf{Game-Theoretic} \\ \textbf{Analysis}\end{tabular}} 
& Existence \& Stability of Equilibria & 
$\bullet$ A Nash Equilibrium (\textbf{NE}) always exists. \newline
$\bullet$ Price of Stability (\textbf{PoS}) is 1\newline
$\bullet$ Price of Anarchy (\textbf{PoA}) is unbounded. \\
\cmidrule(l){2-3}
& Finding a Nash Equilibrium & 
$\bullet$ \textbf{Poly-time} for polynomially bounded edge weights. \newline
$\bullet$ \textbf{PLS-complete} for general weights. \\
\midrule
\multirow{4}{*}{\begin{tabular}[c]{@{}l@{}}\textbf{Optimization} \\ \textbf{(Min-Cost Solution)}\end{tabular}} 
& General Graphs & \textbf{NP-complete} (Generalizes Steiner Orientation). \\
\cmidrule(l){2-3}
& \textbf{Parameterized Algorithms} & \\
& \quad w.r.t. \# of terminal pairs $k$ & \textbf{XP} (\textit{Tight due to W[1]-hardness of Steiner Orientation.}) \\
& \quad w.r.t. \# of edges $|E|$ & \textbf{FPT} \\
& \quad w.r.t. \# of arcs $|A|$ + $k$ &  \textbf{FPT}\\
& \quad w.r.t. \# of arcs $|A|$ + $diam(G)$ &  \textbf{FPT}\\
& \quad w.r.t. vertex cover number $\vc$ + $k$ & \textbf{FPT} (Unweighted case)\\
& \quad w.r.t. vertex cover number $\vc$  & \textbf{FPT} (Unweighted, no overlapping terminals case)\\
& \quad w.r.t. the size of $(V\times V)\setminus \mathcal{T}$ & \textbf{FPT} (Computing zero cost solution)\\
\bottomrule
\end{tabularx}
\end{table*}

\subsection{Our contribution}
We investigate \ccMAR{} model from both a game-theoretic and an algorithmic perspectives.
Table \ref{tab:contributions} summarizes our contributions.

From a game-theoretic perspective, we analyze the strategic behavior of agents who selfishly choose their own routes to minimize their individual crossing costs. 
In this setting, we first establish that the game always has a (pure) Nash equilibrium and prove that best-response dynamics are guaranteed to converge in at most $w_{\max}k^2m$ steps, where $w_{\max}$ is the maximum weight among edges. This result stems from our proof that any strategy profile with minimum cost is a Nash equilibrium, which implies that the Price of Stability (PoS) is 1. In stark contrast, we demonstrate an instance where the Price of Anarchy (PoA) is unbounded. 
Next, we consider the complexity of finding a Nash equilibrium, showing that it can be solved in polynomial time if the weights are bounded by a polynomial in the input size, but becomes PLS-complete otherwise.

Shifting to a centralized optimization perspective, we study the problem of finding a minimum-cost strategy profile, referred to as an optimal solution.
We first observe that our problem generalizes the Steiner Orientation problem. Specifically, deciding whether a zero-cost solution exists is equivalent to solving an instance of the Steiner Orientation problem. Since Steiner Orientation is known to be NP-complete in very restricted graph classes~\cite{hanaka2025structuralparameterssteinerorientation}, so is finding an optimal solution for CC-MAR. To circumvent this hardness, we employ the framework of parameterized complexity. 

Our analysis yields a comprehensive suite of parameterized algorithms. We first present an XP algorithm parameterized by the number of terminal pairs $k$. We next develop two FPT algorithms: one parameterized by the number of edges, and another by the number of arcs plus the number of terminal pairs. 
We then proceed by considering more restricted instances, with particular interest in unweighted graphs. Here, we establish the fixed-parameter tractability for the combined parameter of vertex cover number plus the number of terminal pairs. We also deal with even more restricted cases, such as instances without overlapping terminal vertices.

Crucially, many of our algorithmic results are provably tight, as they are constrained by existing hardness results for Steiner Orientation.
The W[1]-hardness of Steiner Orientation parameterized by the number of terminal pairs~\cite{ChitnisFS19} implies that our XP algorithm is unlikely to be improved to FPT. 
Additionally, since Steiner Orientation is W[1]-hard when parameterized by the vertex cover number~\cite{hanaka2025structuralparameterssteinerorientation}, we do not expect to achieve fixed-parameter tractability with respect to the vertex cover number alone. 

Taken together, our results provide both deep algorithmic insights and a structural understanding of multi-agent routing under asymmetric interaction costs, contributing to the design of more robust and scalable, decentralized coordination mechanisms for systems where strict synchronization is infeasible.

\subsection{Related work}
Our model is at the intersection of Multi-Agent Path Finding, Congestion Games, and Graph Orientation. In this section, we introduce related work in these lines.

\paragraph{Multi-Agent Path Finding (MAPF)}
Multi-Agent Path Finding (MAPF) has received significant attention from both theoretical and practical perspectives, with many real-world applications, e.g., in warehouse management, autonomous vehicle coordination, and multi-robot systems~\cite{DBLP:journals/aim/WurmanDM08,socs/SternSFK0WLA0KB19,DBLP:conf/aaai/0001TKDKK21,DBLP:conf/aaai/000123,DBLP:conf/aaai/0001T23,DBLP:conf/aaai/Yan025}. A significant body of work aims to overcome the computational hardness of MAPF by considering restricted instances and parameterized complexity~\cite{FioravantesKKMO25a,FioravantesKKMO24,DeligkasEGK025,EibenGK025}.
Additionally, several variations of MAPF have been introduced over the years addressing concepts like energy consumption~\cite{DeligkasEGK024}, communication constraints~\cite{TateoBRAB18,FioravantesKKMO25}, rescheduling after delays~\cite{KottingerGASL24}, and many others~\cite{SternSFK0WLA0KB19,AliY24,MaTSKK16,ZhouZR25}.
In contrast to classical MAPF, which models synchronous systems with hard collision constraints, our work focuses on \emph{asynchronous} systems with a \emph{soft} cost function. As such, CC-MAR provides a more flexible and arguably more realistic model for decentralized systems where timing is uncertain.

\paragraph{Congestion Games}

The Congestion Game is first introduced by Rosenthal~\cite{Rosenthal1973}. In a standard congestion game, players select from a set of shared resources, and the cost associated with each resource depends on the number of players using it. Rosenthal proves that every congestion game admits a pure-strategy Nash equilibrium by employing a potential function argument.

The CC-MAR model is most closely related to Network Congestion Games, where graph edges act as the resources. In this setting, the strategy for a player with a terminal pair $(s,t)$ is to choose an $s$-$t$ path. The total cost for a player is the sum of the costs of the edges in their chosen path, where an edge's cost is a non-decreasing function of the number of agents using it. A Nash equilibrium can be computed in polynomial time if all players share the same terminal pair~\cite{stoc/FabrikantPT04}, but finding one is PLS-complete in the general case~\cite{stoc/FabrikantPT04,AckermannRV08}.

\paragraph{Steiner Orientation}
The zero-cost version of CC-MAR is equivalent to the Steiner Orientation problem, where one must orient undirected edges in a mixed graph to connect a set of terminal pairs. This problem is NP-complete, even on graphs of treewidth 2~\cite{ArkinH02, hanaka2025structuralparameterssteinerorientation}.
From a parameterized perspective, its complexity is well-understood: the problem is FPT when parameterized by the number of edges and arcs, but it is W[1]-hard (though in XP) when parameterized by the number of terminal pairs or the vertex cover number~\cite{CyganKN13, DBLP:journals/algorithmica/ChitnisFS19, hanaka2025structuralparameterssteinerorientation}. These hardness results for Steiner Orientation apply directly to CC-MAR.

\section{Preliminaries}\label{sec:prelim}
\paragraph*{Notations and Problem Setting.}
Let $G=(V,E,A)$ be a mixed graph, where $E\subseteq \binom{V}{2}$ denotes the set of (undirected) edges in $G$ and $A\subseteq V\times V$ denotes the set of directed edges in $G$. An undirected edge is called just an \emph{edge}, and a directed edge is called an \emph{arc}. 
If each edge $e\in E$ has the weight $w_e$, the graph $G$ is called a \emph{weighted graph}. Throughout the paper, we assume that the edge weights are positive integers.

A sequence of vertices $v_1-\ldots-v_m$ is a path in $G$ if $i,j \in [m]$ if $i\neq j$ then $v_i \neq v_j$, and for any $i \in [m-1]$, either $\{v_i,v_{i+1}\}\in E$ or $(v_i,v_{i+1})\in A$.

We say that paths $P_1$ and $P_2$ use an edge $\{u,v\}\in E$ in opposite directions if $u$ and $v$ appear consecutively in both paths and $u$ precedes $v$ in $P_1$ if and only if $v$ precedes $u$ in $P_2$. We define the crossing cost $\cost(P_1,P_2)$ as the sum of weights $w_{uv}$ of edges $\{u,v\}\in E$ that are used in opposite directions by $P_1$ and $P_2$.

 For each edge $e = \{u,v\}\in E$, we define $\overrightarrow{e}=(u,v)$ and $\overleftarrow{e}=(v,u)$. Let
    $x_{\overrightarrow{e}}$ (resp., $x_{\overleftarrow{e}}$)  be the number of agents on $\overrightarrow{e}$ (resp., $\overleftarrow{e}$) in $\PS$. 
Without loss of generality, we assume that a path $P$ uses $\overrightarrow{e}=(u,v)$ if it uses edge $e=\{u,v\}$.
We denote by $E(P)$ the set of edges used by a path $P$.

Given a set of paths $\PS=\{P_1,\ldots,P_k\}$,
the cost of a path $P_i$ is defined as: 
\[\cost_{\PS}(P_i) = \sum_{e\in E(P_i)}w_{e}x_{\overleftarrow{e}}\]
The total cost of $\PS$ is defined as: %
\begin{align*}
 \cost(\PS) & = \frac{1}{2}\sum_{i=1}^k\sum_{e\in E(P_i)}w_{e}x_{\overleftarrow{e}} = \sum_{e\in E} w_e x_{\overrightarrow{e}}x_{\overleftarrow{e}} \\
 & = \sum_{1\le i< j\le k} \cost(P_i,P_j) = \sum_{\{u,v\}\in E}w_{uv}x_{uv}x_{vu}   
\end{align*}
\noindent
Our model is formally defined as follows.

\defmodel{Crossing Cost Multi-Agent Routing (\ccMAR )}
{$\Gamma = (G, \terminals)$ where $G=(V,E,A)$ is a (weighted) mixed graph,  $\terminals$ is a multiset that contains $k$ pairs of vertices $(s_i,t_i) \in V\times V$, $i \in [k]$.}
{The set $\PS=\{P_1,\ldots,P_k\}$ of $k$ paths $P_i = s_i-\ldots-t_i$, $i \in k$.}
{$\cost_{\PS}(P_i) = \sum_{e\in E(P_i)}w_{e}x_{\overleftarrow{e}}$}
{$\cost(\PS) =\sum_{\{u,v\}\in E}w_{uv}x_{uv}x_{vu}$}

We say $\Gamma = (G, \terminals)$ is \emph{feasible} if there exists at least one $s_i$-$t_i$ path in $G$ for each $i\in [k]$. In this paper, we assume that $\Gamma$ is feasible because we can check whether $\Gamma$ is feasible in polynomial time by computing an $s_i$-$t_i$ path for $i\in [k]$ on the directed graph obtained by replacing each edge $\{u,v\}$ in $E$ by directed arcs $(u,v)$ and $(v,u)$.
When we deal with the optimization version of \ccMAR, a strategy profile will be called a feasible solution, while the feasible solution of minimum total cost will be called an optimal solution.


\paragraph*{Parametrized Complexity.}
As we have already mentioned, we will result to parameterized complexity in order to overcome the tractability hardness of computing optimal solutions of \ccMAR. 
Parameterized complexity extends classical complexity analysis by introducing a secondary metric called the parameter. Under this framework, the problems are defined as pairs $(x,k) \in \Sigma^* \times \mathbb{N}$, where $k$ is the parameter.
The goal is to capture all the computational hardness in the parameter and provide algorithms that scale well as the input grows. 
In particular, a problem admits a fixed-parameter tractable algorithm (FPT) if it is solvable in $f(k)|x|^{O(1)}$ time (for a computable function $f$). Similarly, a problem admits a slicewise polynomial (XP) algorithm if it is solvable in $|x|^{f(k)}$ time. Problems that admit FPT algorithms belong to the complexity class FPT while problems that admit XP algorithms belong to the XP class.
Under standard computational complexity assumptions problems proven W[1]-hard (via parameterized reduction) do not belong to FPT while problems proven ParaNP-hard (via parameterized reduction) do not belong to XP. 
We refer to the monographs~\cite{CyganFKLMPPS15,downey2012parameterized,fomin2019kernelization} for a dedicated study of parameterized complexity.

\section{Game Theoretic Results}

In this section, we investigate Nash Equilibria on CC-MAR.

\subsection{Existence and convergence of a Nash equilibrium}
\begin{theorem}\label{theorem:Nash}
    A Nash equilibrium always exists, and the Nash dynamics converge with at most $w_{\max}k^2|E|$ steps,  where $w_{\max}$ is the maximum weight among edges.
\end{theorem}

\begin{proof}
    We show that $\cost(\PS)$ is a valid potential function. 
    Let $\PS$ be the set of $s_i$-$t_i$ paths $P_1, \ldots, P_k$.  
    For each edge $e = \{u,v\}\in E$, we define $\overrightarrow{e}=(u,v)$ and $\overleftarrow{e}=(v,u)$. Let
    $x_{\overrightarrow{e}}$ (resp., $x_{\overleftarrow{e}}$)  be the number of agents on $\overrightarrow{e}$ (resp., $\overleftarrow{e}$) in $\PS$. 
    Then the cost of $\PS$ is denoted by $\cost(\PS)=\sum_{e\in E} w_e x_{\overrightarrow{e}}x_{\overleftarrow{e}}$. 
    
    If $\PS$ is not a Nash equilibrium, there exists an agent $i$ that wants to change its path $P_i$ to a path $P'_i$ whose cost is less than $P_i$. Let $\PS'=\{P_1, \ldots, P'_i, \ldots, P_k\}$. Without loss of generality, we assume that whenever $P_i$ uses an edge ${u,v}\in E$, it does so in the direction $\overrightarrow{e}=(u,v)$. We categorize edges with respect to the difference between $P_i$ and $P'_i$:
    \begin{comment}
    \begin{itemize}
        \item 
        $E(P_i)\setminus E(P'_i)$ denotes the set of edges used by only $P_i$.  Without loss of generality, we suppose that $P_i$ uses ${\overrightarrow{e}}$.
        \item 
        $E(P'_i)\setminus E(P_i)$ denotes the set of edges used by only $P'_i$.  Without loss of generality, we suppose that $P'_i$ uses ${\overleftarrow{e}}$.
        \item $F_i$ denotes the set of edges used by both $P_i$ and $P'_i$ in the same direction. Without loss of generality, we suppose that both $P_i$ and $P'_i$ use ${\overrightarrow{e}}$.
        \item $R_i$ denotes the set of edges used by both $P_i$ and $P'_i$ in the different directions. Without loss of generality, we suppose that $P_i$ uses ${\overrightarrow{e}}$ and $P'_i$ uses ${\overleftarrow{e}}$.
        \item $E\setminus (E(P_i)\cup E(P'_i))$ denotes the set of edges used by neither $P_i$ nor $P'_i$.
    \end{itemize}
\end{comment}
   \begin{itemize}
        \item $E_{\overrightarrow{i}}$ denotes the set of edges used by only $P_i$, i.e., $E(P_i)\setminus E(P'_i)$.  Without loss of generality, suppose $P_i$ uses ${\overrightarrow{e}}$.
        \item $E_{\overleftarrow{i}}$ denotes the set of edges used by only $P'_i$, i.e., $E(P'_i)\setminus E(P_i)$.  Also, suppose $P'_i$ uses ${\overleftarrow{e}}$.
        \item $F_i$ denotes the set of edges used by both $P_i$ and $P'_i$ in the same direction. Also, suppose both $P_i$ and $P'_i$ use ${\overrightarrow{e}}$.
        \item $R_i$ denotes the set of edges used by both $P_i$ and $P'_i$ in the different directions. Also, suppose $P_i$ uses ${\overrightarrow{e}}$ and $P'_i$ uses ${\overleftarrow{e}}$.
        \item $\hat{E}_i $ denotes the set of edges used by neither $P_i$ nor $P'_i$, i.e., $E\setminus (E(P_i)\cup E(P'_i))$.
    \end{itemize}
Then the total cost of $\PS$ can be represented as follows:
\begin{align*}
\cost(\PS)=
        \sum_{e\in E_{\overrightarrow{i}}\uplus E_{\overleftarrow{i}}
        \uplus F_i \uplus R_i \uplus \hat{E}_i
        } w_e x_{\overrightarrow{e}}x_{\overleftarrow{e}}.
\end{align*}
Thus, we have:
       \begin{align*}
    \cost(\PS') - \cost(\PS)  \begin{multlined}[t]
         \sum_{e\in E_{\overleftarrow{i}}} w_e x_{\overrightarrow{e}}
         -\sum_{e\in E_{\overrightarrow{i}}} w_e x_{\overleftarrow{e}} 
         + \sum_{e\in R_i} w_e (x_{\overrightarrow{e}}-x_{\overleftarrow{e}} - 1).
       \end{multlined}
\end{align*}

On the other hand, since
$$\cost_{\PS}(P_i)=  \sum_{e\in E_{\overrightarrow{i}}} w_e x_{\overleftarrow{e}} x_{\overrightarrow{e}}
 + \sum_{e\in F_i} w_e x_{\overleftarrow{e}}
+ \sum_{e\in R_i} w_e x_{\overleftarrow{e}},$$
we have:
\begin{align*}
    \cost_{\PS'}(P'_i) - \cost_{\PS}(P_i)
    &= \begin{multlined}[t]
         \sum_{e\in E_{\overleftarrow{i}}} w_e x_{\overrightarrow{e}} - \sum_{e\in E_{\overrightarrow{i}}} w_e x_{\overleftarrow{e}} 
         + \sum_{e\in R_i} w_e (x_{\overrightarrow{e}} - x_{\overleftarrow{e}} - 1)
       \end{multlined}\\
    &= \cost(\PS') - \cost(\PS).
\end{align*}

    
    By assumption that $\cost_{\PS'}(P'_i)  - \cost_{\PS}(P_i) < 0$, $\cost(\PS') - \cost(\PS)<0$ holds, implying that $\cost(\PS)$ is the potential function. Since the maximum cost of a strategy profile is at most $w_{\max}k^2|E|$, the Nash dynamics converge at most $w_{\max}k^2|E|$ steps.
\end{proof}

\subsection{Price of anarchy and price of stability}
\begin{theorem}\label{theorem:pos}
    The price of stability is 1.
\end{theorem}
\begin{proof}\quad 
    By the proof of \Cref{theorem:Nash}, the minimum cost solution is a Nash equilibrium.
\end{proof}

\begin{theorem}\label{theorem:poa}
    The price of anarchy is unbounded even for monotone edge costs. 
\end{theorem}
\begin{proof}
    Let $G=(\{a,b,c\}, \{\{a,b\},\{b,c\},\{c,a\}\}$ be a complete graph with three vertices.  There are two agents on each vertex and we denote them by $A_1$, $A_2$, $B_1$, $B_2$, $C_1$, $C_2$.
    The terminal vertices of $A_1$, $A_2$, $B_1$, $B_2$, $C_1$, $C_2$ are $b$, $c$, $c$, $a$, $a$, $b$, respectively.

    Let $\PS_1$ be a strategy profile as follows:
     $P_{A_1} = a-b$, $P_{A_2} = a-c$,   $P_{B_1} = b-c$, $P_{B_2} = b-a$, $P_{C_1} = c-a$, and $P_{C_2} = c-b$.
    Under $\PS_1$, the cost of each agent is 1. Moreover, it is easily seen that each agent cannot reduce its cost by changing the current path. Thus, $\PS_1$ is a Nash equilibrium, and the cost of $\PS_1$ is 6.

    Next, consider the following strategy profile $\PS_2$ as follows:
    $P_{A_1} = a-b$, $P_{A_2} = a-b-c$, $P_{B_1} = b-c$, $P_{B_2} = b-c-a$, $P_{C_1} = c-a$, $P_{C_2} = c-a-b$.
    Then the cost of  $\PS_2$ is 0. Therefore, the price of anarchy is unbounded.
\end{proof}

\subsection{Computing a Nash equilibrium}
\begin{theorem}\label{thm:Nash:verif}
    Given a strategy profile $\PS$, one can determine whether $\PS$ is a Nash equilibrium in $O(k(|E|\log w_{\max} + |V|\log |V|))$. Furthermore, if $\PS$ is not a Nash equilibrium, an improving path can be found in the same running time.
\end{theorem}
\begin{proof}
Let $G=(V,E,A)$ be an input graph. Given a strategy profile $\PS$, we verify for each agent $i$ whether it has an incentive to change the current path $P_i$. To do this, we define the directed graph $G'=(V, A')$ obtained from $G$ by replacing (undirected) edges in $E$ with bidirectional arcs. For each arc $(u,v)\in A'$, we set the weight $w'_{uv}$ as follows:
    \begin{align*}
    w'_{uv} = 
        \begin{cases}
            w_{uv} (x_{vu}-1)& \text{(if $vu\in P_i$ and $\{u,v\}\in E$)}\\
            w_{uv} x_{vu} & \text{(if $vu\notin P_i$ and $\{u,v\}\in E$)}\\
            0 & \text{(otherwise)}
        \end{cases}
    \end{align*}
    Let $P'_i$ be a shortest $s_i$-$t_i$ path in $G'$.
    We can see that the weight of the path $P'_i$ in $G'$ is equal to the cost of the path $P'_i$ for agent $i$ when it changes its strategy to $P'_i$. Therefore, if the weight of $P_i'$ in $G'$ is less than the cost $\cost_{\PS}(P_i)$ of the path $P_i$ under $\PS$, agent $i$ has an incentive to change, and $P'_i$ is an improving path for agent $i$. On the other hand, if there is no such path, any other $s_i$-$t_i$ path is at least as costly as $P_i$ when the paths of other agents are fixed, and hence agent $i$ has no incentive to change its path. 
    If any agent has no improving path, then $\mathcal{P}$ is a Nash equilibrium.
    Since the edge weights can be set in time $O(|E| \log w_{\max})$ and a shortest $s_i$-$t_i$ path can be computed in time $O(|E| + |V|\log |V|)$, one can verify whether $\PS$ is a Nash equilibrium in time $O(k(|E|\log w_{\max} + |V|\log |V|))$.
\end{proof}

From \Cref{theorem:Nash} and \Cref{thm:Nash:verif}, we obtain the following corollary.

\begin{corollary}
    One can find a Nash equilibrium in time \break $O(w_{\max}k^2|E| (|E|\log w_{\max} + |V|\log |V|))$.
\end{corollary}

At the end of this section, we prove that finding a Nash equilibrium is PLS-complete by showing a PLS-reduction from the Quadratic Threshold Game (QTG), which is PLS-complete~\cite{AckermannRV08}. 
\begin{theorem}\label{thm:Nash:PLS}
    Finding a Nash equilibrium in CC-MAR is PLS-complete.
\end{theorem}

From \Cref{thm:Nash:verif}, finding a Nash equilibrium in weighted CC-MAR is in PLS.
To show the PLS-hardness, we introduce the Quadratic Threshold Game (QTG), which is a special case of congestion games.

Let $B=\{1,\ldots, n\}$ be the set of $n$ agents.
Let $R=R^{\mathrm{in}}\cup R^{\mathrm{out}}$ be a set, called a \emph{resource set} where $R^{\mathrm{in}}\cap R^{\mathrm{out}} = \emptyset$, and $R^{\mathrm{in}}=\{r_{ij} \mid i,j\in B\}$ and $R^{\mathrm{out}}=\{r_i \mid i\in B\}$. That is, resource $r_i$ is used by only agent $i$. If $i$ uses $r_i$, it takes cost $T_i$, called \emph{threshold}. 
An element $r_{ij} \in R^{\mathrm{in}}$ is a resource for $i,j\in B$. Thus, $r_{ij}$ is used by $i$ or $j$, or both. The cost of $r_{ij}$ is defined by a nondecreasing function $c_{ij}:\{0,1,2\} \to \mathbb{Z}_{\ge 0}$ depending on the number of agents $x_{r_{ij}}\in \{0,1,2\}$ using $r_{ij}$. Here, $\mathbb{Z}_{\ge 0}$ denotes the set of nonnegative integers. 

Each agent $i$ has two strategies $S_{i}^{\mathrm{out}} = \{r_{i}\}$ where $r_{i} \in R^{\mathrm{out}}$ and $S_{i}^{\mathrm{in}} = \{r_{ij}\mid r_{ij} \in R^{\mathrm{in}}, j\in B, i\neq j\}$.
By the above definition, when $i$ chooses strategy $S_{i}^{\mathrm{out}}$, then the cost of $i$ is $T_{i}$, and when $i$ chooses strategy $S_{i}^{\mathrm{in}}$, then the cost of $i$ is $\sum_{j\in B, j \neq i} c_{ij}(x_{r_{ij}})$.
Therefore, a quadratic threshold game is denoted by $\Delta=(B,\{T_i\mid i\in B\},\{c_{ij}\mid i,j\in B, i\neq j\})$

Here, every agent can see other agents' strategies and prefers a less expensive strategy. For example, if the cost of choosing $S_{i}^{\mathrm{in}}$ is less than the current cost of choosing $S_{i}^{\mathrm{out}}$, then agent $i$ has an incentive to change its strategy. If any agent has no incentive to change, such a strategy profile is a Nash equilibrium.

\cite{AckermannRV08} shows that a Nash equilibrium on quadratic threshold games is PLS-complete.
\begin{theorem}[\cite{AckermannRV08}]
Finding a Nash equilibrium on quadratic threshold games with non-decreasing cost functions is PLS-complete under the following assumption.
\begin{itemize}
 \item for each pair $i$ and $j$, the cost function $c_{ij}$ satisfies $c_{ij}(0) = c_{ij}(1) = 0$ and $c_{ij}(2)\in \mathbb{Z}_{\ge 0}$, and
 \item for each $i$, the threshold $T_{i}$ is positive. 
\end{itemize}
\end{theorem}

Thus, it is sufficient to show that any quadratic threshold game $\Delta=(B,\{T_i\mid i\in B\},\{c_{ij} \mid i,j\in B\})$ can be PLS-reduced to some instance of a CC-MAR $\Gamma = (G,\terminals)$. Our reduction presented here modifies a part of the reduction from QTG to a network congestion game in \cite{AckermannRV08}. 

\medskip

\noindent \textit{Proof of \Cref{thm:Nash:PLS}.}
Given a QTG $\Delta=(B,\{T_i\mid i\in B\},\{c_{ij} \mid i,j\in B\})$ with $n=|B|$ agents, we construct a grid-like graph as shown in \Cref{fig:grid}. This graph consists of a grid section with vertices arranged on the lower-left triangle of an $n \times n$ grid, including diagonal vertices where vertical arcs are oriented downwards and horizontal arcs are oriented from left to right. We denote by $v_{ij}$ the vertex placed in $i$th row and $j$th column.
Moreover, the graph has arcs from $v_{i1}$ to $v_{ni}$ for each $i\in [n]$. We call the arc $(v_{i1},v_{ni})$ the \emph{threshold arcs} of agent $i$. 

\begin{figure}[tb]
\centering
\includegraphics[keepaspectratio,width=0.4\linewidth]{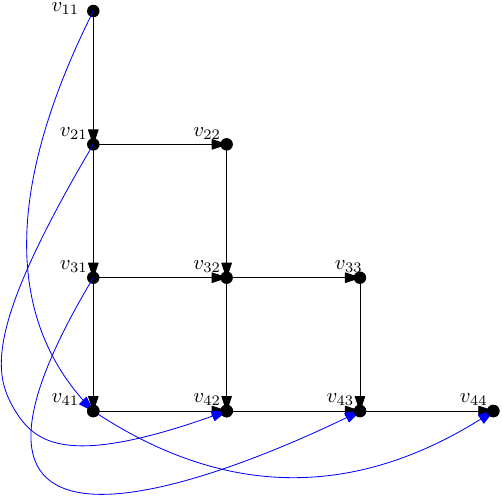}
\caption{A grid-like graph when $n=4$. The blue arcs are threshold arcs.} \label{fig:grid}
\end{figure}

We then define $n$ agents $p_1, \ldots, p_n$, called \emph{main agents}.  Their terminal pairs are given $(s_i, t_i) = (v_{i1}, v_{ni})$ for $i\in [n]$. Intuitively, each agent $i$ has two strategy: the threshold arc $(s_i, t_i)$ of $i$ and the \emph{row-column path} from $s_i$ to $t_i$, which first passes along horizontal arcs of row $i$ until column $i$ and then goes down to $t_i$ along vertical arcs in column $i$. The selection of these two paths corresponds to the selection of $S_{i}^{\mathrm{out}}$ and $S_{i}^{\mathrm{in}}$ in the QTC.

To enforce each agent to choose either its threshold arc or its row-column path, we additionally modify the instance.
First, we replace each horizontal arc and threshold arc $(u,v)$ by a mixed $u$-$v$ path consisting of $(u,u'), \{u',v'\},(v',v)$ as shown in \Cref{fig:replace:arc}.
Let $D$ be a large integer.
If an arc $(u,v)$ is a horizontal arc in $i$th row, we define the weight of edge $\{u',v'\}$ as  $w_{u'v'} = D\cdot i$. If an arc $(u,v)$ is the threshold arc,  we define the weight of edge $\{u',v'\}$ as  $w_{u'v'} = D\cdot i \cdot (i-1) + T_i$.
Moreover, we introduce an agent $a_{v'u'}$ with its terminal pair $(v',u')$ for each edge $\{u',v'\}$. Note that by the construction, the agent $a_{v'u'}$ always chooses the path $\langle v',u'\rangle$ as its strategy.

\begin{figure}[tb]
\centering
\includegraphics[keepaspectratio,width=0.35\linewidth]{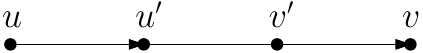}
\caption{The replacement of arc $(u,v)$ by a mixed $u$-$v$ path consisting of $(u,u'), \{u',v'\},(v',v)$.} \label{fig:replace:arc}
\end{figure}

\begin{figure}[t]
\centering
\includegraphics[width=0.35\linewidth]{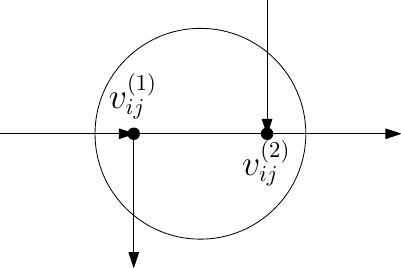}
\caption{The resource gadget of $r_{ij}$ obtained by replacing vertex $v_{ij}$.} \label{fig:replace:vertex}
\end{figure}

Finally, to simulate the cost of the resource $r_{ij}$, for $1\le i<j\le n$, we split vertex $v_{ij}$ into two vertices  $v^{(1)}_{ij}$ and $v^{(2)}_{ij}$ and add edge $\{v^{(1)}_{ij},v^{(2)}_{ij}\}$ with weight $c_{ij}$. Incoming arcs to $v_{ij}$ are incident to $v^{(1)}_{ij}$ and outgoing arcs from $v_{ij}$ are incident to $v^{(2)}_{ij}$.
We then introduce agent $a_{ij}$ whose terminal pair is $(v^{(2)}_{ij}, v^{(1)}_{ij})$. Similarly to $a_{v'u'}$, the strategy of $a_{ij}$ is only the path $\langle v^{(2)}_{ij}, v^{(1)}_{ij}\rangle$. We call this replacement part the \emph{resource gadget} of $r_{ij}$ (see \Cref{fig:replace:vertex}).
As for the terminal pair of agent $p_i$ for $2\le i\le n$, we set $(s_i,t_i) = (v^{(1)}_{i1}, v^{(1)}_{ni})$. The terminal pair of agent $p_1$ is set as $(s_1,t_1) = (v_{11}, v^{(1)}_{n1})$.
Since $D$ is a sufficiently large integer, each agent $p_i$ has only two strategies: the threshold arc $(s_i, t_i)$ of $i$ and the row-column path from $s_i$ to $t_i$, which first passes along horizontal arcs of row $i$ until column $i$ and then goes down to $t_i$ along vertical arcs in column $i$.

 The cost of the threshold path for $i$ is clearly $D \cdot i \cdot (i-1) +T_{i}$. The cost of the row-column path for $i$ is the sum of costs from the $i$ horizontal \emph{undirected} edges and the resource gadgets it traverses, which equals $D\cdot i \cdot (i-1) + \sum_{j \in B, j \neq i}c_{ij}(x_{r_{ij}})$ where $x_{r_{ij}}\in \{0,1,2\}$ represents the number of uses for the resource gadget of $r_{ij}$. 
 Indeed, if agent $i$ chooses the row-column path, it crosses agents in horizontal undirected edges and several main agents who also choose their row-column paths in the resource gadgets.
 Note that the cost function $c_{ij}$ satisfies that $c_{ij}(0) = c_{ij}(1) = 0$ and $c_{ij}(2)\in \mathbb{Z}_{\ge 0}$.
 Consequently, for each agent $i$, its strategy $S_i^{\text{out}}$ in the QTG $\Delta$ can be mapped one-to-one with the threshold arc of $i$ in the weighted CC-MAR $\Gamma$, and its strategy $S_i^{\text{in}}$ corresponds to the row-column path of $i$. Furthermore, we can easily see that a strategy profile in the QTG $\Delta$ is a Nash equilibrium if and only if a strategy profile $\mathcal{P}$ in the weighted CC-MAR $\Gamma$ is a Nash equilibrium. 
 Therefore, our construction is a PLS-reduction from QTG to weighted CC-MAR.
\qed

\section{Minimum-Cost Solution via Parameterization}

Hereafter, we consider the problem of finding a minimum-cost strategy profile (i.e., an optimal solution). 
Before presenting our main results, we establish key observations used throughout this and the next section.


\paragraph*{Cycles}
We can always assume that $G$ has no cycles. 
If this is not the case, we repeatedly contract a mixed cycle $C$ to a single vertex, while such a cycle exists. Finding such cycles can be done in polynomial time using the classic DFS algorithm~\cite{cormen01introduction}. 
Eventually, we obtain the resulting graph $G'=(V',E',A')$ having no mixed cycles. 
This procedure is safe since any agent moving from $u\in C$ to $v\in C$ could move within $C$ without cost. 
Hereafter, we assume that $G$ is acyclic, and therefore, $G[A]$ is a DAG and $G[E]$ is a forest.

\paragraph*{Pendant vertices}

We can always assume that the given graph does not have any pendant vertices i.e., vertices of degree~$1$.
In particular, we show how to deal with pendant vertices, before we delete them. 

Note that any pendant vertex without a terminal is redundant, as it is never used in any path (its usage would result in a walk instead). Therefore, it can be removed.

Now, consider a pendant vertex $v$ that appears as a terminal. Let $u$ be the unique neighbor of $v$ in $G$. If $\{u,v\} \notin E$, then $v$ is incident to an arc. 
Since arcs do not affect the cost of the solution, and a feasible solution exists (by assumption), we can replace all appearances of $v$ in $\terminals$ with $u$ and remove $v$.

Next, we consider the case where $\{u,v\} \in E$. If $v$ appears only as a source or only as a sink terminal, then the edge $\{u,v\}$ is used either only from $v$ to $u$, or only from $u$ to $v$. Therefore, it does not affect the cost, and we can safely replace all appearances of $v$ in $\terminals$ with $u$ and remove $v$.

Now, assume that $v$ appears in $\terminals$ both as a source and as a sink terminal. Let $\ell_s$ and $\ell_t$ be the number of times $v$ appears as a source and as a sink terminal, respectively. Since we assume that no pair $(v,v)\in \terminals$, any feasible solution includes $\ell_s$ paths that use $\{u,v\}$ to go from $v$ to $u$, and $\ell_t$ paths that use $\{u,v\}$ to go from $u$ to $v$. Also, no other path uses the edge $\{u,v\}$. This gives a cost of $\ell_s \cdot \ell_t$. Therefore, we can replace all appearances of $v$ in $\terminals$ with $u$, remove $v$, and adjust the cost appropriately (i.e., either set the cost as $\ell - \ell_s \cdot \ell_t$ in the decision version, or increase the cost by $\ell_s \cdot \ell_t$ after computing a solution of the new instance).

From the above, we establish a key observation.

\begin{observation}[Input Simplification] \label{obs:input:simpl}
    Without loss of generality, the input graph can be assumed to be acyclic and free of degree-$1$ vertices.
\end{observation}

\paragraph*{Identical terminal pairs}
Recall that in the definition, we have allowed the same pair to appear multiple times in $\terminals$. This is important, as having identical terminal pairs multiple times affects the cost and may lead to completely different optimal solutions depending on the number of times each pair appears.

In this part, we prove that regardless of how many times a pair appears, there exists an optimal solution in which all identical pairs use exactly the same path to connect their terminals. To prove this, consider an instance $(G,\terminals)$ and an optimal solution $\mathcal{P}$ such that: 
\begin{itemize}
    \item There exist two terminal pairs $(s_i,t_i) = (s_j,t_j)$, where $0 < i < j \le k$, and 
    \item the $s_i$-$t_i$ and $s_j$-$t_j$ paths in the optimal solution are $P_i, P_j \in \mathcal{P}$ with $P_i \neq P_j$.
\end{itemize}
We will show that we can select one of the paths $P_i$ or $P_j$ and replace both with the selected one.

Let $P = P_i$. Note that $\cost(\mathcal{P} \setminus P_i) = \cost(\mathcal{P} \setminus P_j)$. Indeed, if $\cost(\mathcal{P} \setminus P_i) < \cost(\mathcal{P} \setminus P_j)$ or $\cost(\mathcal{P} \setminus P_i) > \cost(\mathcal{P} \setminus P_j)$, then the feasible solution $(\mathcal{P} \setminus \{P_j\}) \cup \{P_i\}$ or $(\mathcal{P} \setminus \{P_i\}) \cup \{P_j\}$ would have smaller cost than $\mathcal{P}$, contradicting optimality.

Finally, note that $(\mathcal{P} \setminus \{P_j\}) \cup \{P_i\}$ must have the same cost as $\mathcal{P}$, and therefore, there exists a solution in which $(s_i,t_i)$ and $(s_j,t_j)$ are connected by identical paths.

We can easily generalize this to any number of identical pairs by repeating the same process several times.


\begin{observation}[Solution Structure]\label{obs:solution:stucture}
    For any \ccMAR{} instance $(G,\terminals)$, there exists a minimum-cost solution where identical terminal pairs are connected by identical paths.
\end{observation}

\paragraph*{Parameterization by the number of vertices}
We present a fixed-parameter tractable (FPT) algorithm parameterized by the number of vertices $|V|$. This algorithm serves as a subroutine for several methods introduced later in the paper.
\begin{proposition}\label{prop:FPT:by:graph:size}
    A minimum-cost solution for a \ccMAR{} instance $(G,\terminals)$ can be computed in time $|V|^{O(|V|^3)} k^{O(1)}$.
\end{proposition}

\begin{proof}
   The set $\terminals$ is a multiset. Let $\terminals'$ be the set of terminal pairs obtained from $\terminals$ after removing duplicates. Let $(s'_i,t'_i)$ be an enumeration of the terminal pairs in $\terminals'$, and let $k_i$ be the number of times $(s'_i,t'_i)$ appears in $\terminals$. Note that $k_i \ge 1$ for all $i \in |\terminals'|$.

    Recall that \Cref{obs:solution:stucture} showed that there exists an optimal solution of $(G,\terminals)$ where, for all $i \in |\terminals'|$, all $k_i$ copies of $(s'_i,t'_i)$ use exactly the same path.

    Note that, if we are given a path $P_i$ for each $i \in |\terminals'|$ such that $P_i$ is an $s'_i$-$t'_i$ path in $G$, we can compute the cost of a feasible solution where only these paths are used to connect the terminal pairs, in polynomial time. Indeed, the cost of such a feasible solution is:
    \[
    \sum_{0<i<j\le |\terminals'|} k_i \cdot k_j \cdot \cost(P_i, P_j)
    \]
    where $\cost(P_i, P_j)$ can be computed in polynomial time, and $|\terminals'| \le |V(G)|^2$. Therefore, we just need to compute which paths are used by each terminal pair in $\terminals'$.

    Fix a pair $(s',t') \in \terminals'$. Note that there exist at most $2^{|V(G)|} \cdot |V(G)|!$ different paths between $s'$ and $t'$. Indeed, we can enumerate all of them by guessing which vertices belong to $P_i$ and the order in which they appear. This also gives us the upper bound on the number of paths.

    Therefore, in order to compute an optimal solution of $(G,\terminals)$, it suffices to:
    \begin{itemize}
        \item guess, for each pair $(s'_i,t'_i) \in \terminals'$, the path that its copies will use,
        \item compute the cost of the guessed paths, and, 
        \item return the guess of minimum cost.
    \end{itemize}
    Since we have at most $|V(G)|^2$ pairs in $\terminals'$, the number of guesses is bounded by $|V(G)|^{O(|V(G)|^3)}$. Since the rest of the computations can be done in polynomial time, the algorithm runs in time $|V(G)|^{O(|V(G)|^3)} (|V(G)| + k)^{O(1)}$.
\end{proof}

\subsection{Parameterized by the number of agents}

\begin{theorem}\label{thm:XP:k}
    An minimum-cost solution for a \ccMAR{} instance $(G,\terminals)$ can be computed in time $|V|^{O(k^2)}k^{O(1)}$.
\end{theorem}

\begin{proof}
The strategy is as follows: (1) We prove a structural property that in an optimal solution, any pair of agents cross each other on at most one undirected path, (2) We guess these crossing interactions for all $O(k^2)$ pairs of agents, and (3) For each guess, we reduce the problem of finding the remaining non-crossing paths to \textsc{Steiner Orientation}.

    \begin{claim}\label{claim:onepath:cross}
    In a minimum-cost solution $\PS^*$, any pair of agents crosses on at most one undirected path within $G[E]$.
    \end{claim}

    \begin{proofclaim}
        Suppose for contradiction that agents $i$ and $j$ cross on at least two disjoint undirected paths, $Q_1$ and $Q_2$, in $G[E]$. Let $P_i$ and $P_j$ be their respective paths in $\PS^*$. Assume that along $P_i$, the crossing on $Q_1$ occurs before the crossing on $Q_2$.
        Let $(u_1,v_1)$ and $(u_2,v_2)$ be the endpoints of $Q_1$ and $Q_2$, respectively. Without loss of generality, we assume $P_i$ traverses $Q_1$ from $u_1$ to $v_1$ and $Q_2$ from $u_2$ to $v_2$. Thus, $P_i$ has the structure $s_i \rightsquigarrow u_1  \rightsquigarrow v_1 \rightsquigarrow u_2  \rightsquigarrow v_2 \rightsquigarrow t_i$. Correspondingly, $P_j$ must traverse $Q_1$ from $v_1$ to $u_1$ and $Q_2$ from $v_2$ to $u_2$.

        There are two fundamental ways for $P_j$ to connect these segments, as shown in \Cref{fig:crosspaths:param_k}.

        \begin{figure}
            \centering
            \includegraphics[width=0.5\linewidth]{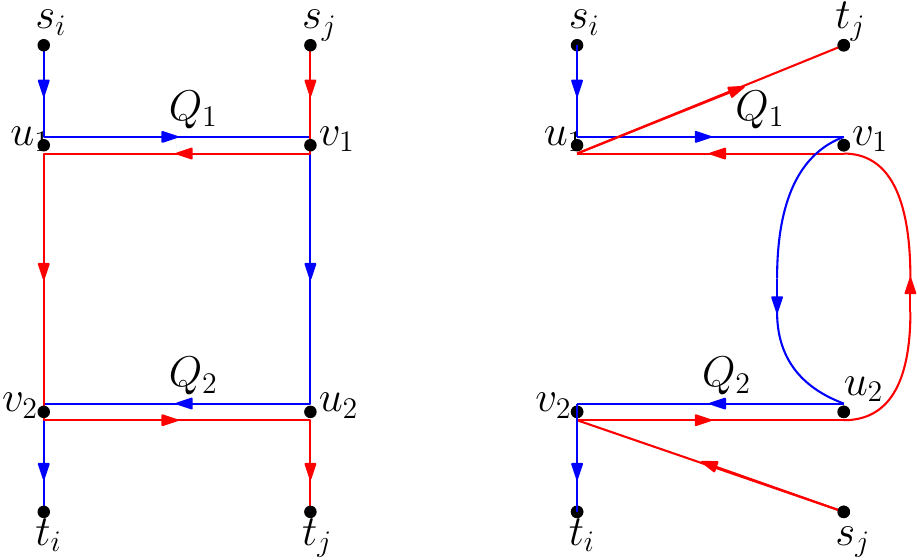}
            \caption{Case analysis in Claim~\ref{claim:onepath:cross}.}
            \label{fig:crosspaths:param_k}
        \end{figure}

        \textbf{Case 1:} $P_j$ is of the form $s_j \rightsquigarrow v_1  \rightsquigarrow u_1 \rightsquigarrow v_2  \rightsquigarrow u_2 \rightsquigarrow t_j$. Let $P'_i$ be the directed subpath of $P_i$ from $v_1$ to $u_2$, and $P'_j$ be the directed subpath of $P_j$ from $u_1$ to $v_2$. We can reroute the agents by swapping these directed subpaths, creating new paths $P_i^* = s_i\rightsquigarrow u_1 \overset{\text{$P'_j$}}{\rightsquigarrow} v_2 \rightsquigarrow t_i$ and $P_j^* = s_j \rightsquigarrow v_1  \overset{\text{$P'_i$}}{\rightsquigarrow} u_2 \rightsquigarrow t_j$. In this new solution, agents $i$ and $j$ no longer cross on $Q_1$ or $Q_2$. Since the rerouting occurs only on directed subpaths (where no costs are incurred), no new costs are introduced, but the two original crossing costs on $Q_1$ and $Q_2$ are eliminated. This yields a new solution with a strictly lower total cost, which contradicts the optimality of $\PS^*$.

        \textbf{Case 2:} $P_j$ is of the form $s_j \rightsquigarrow v_2 \rightsquigarrow u_2 \rightsquigarrow v_1 \rightsquigarrow u_1 \rightsquigarrow t_j$. In this case, the path segments form a mixed cycle $v_1 \rightsquigarrow u_2 \rightsquigarrow v_1$, which contradicts a standard assumption that the graph has no mixed cycles.
        
        Therefore, any pair of agents crosses on at most one undirected path in $G[E]$.
    \end{proofclaim}

    Based on \Cref{claim:onepath:cross}, our algorithm can guess all crossing interactions. For each of the $\binom{k}{2} = O(k^2)$ pairs of agents, we guess (a) whether they cross, and (b) if so, on which path in $G[E]$ they cross. Since $G[E]$ is a forest, a path is uniquely defined by its endpoints. The number of possible paths to guess for each agent pair is bounded by $O(|V|^2)$. Thus, the total number of guesses is $|V|^{O(k^2)}$.

    For a fixed guess of all crossing interactions, we must determine the full path for each agent. This involves not only the specified undirected crossing paths but also the order in which each agent traverses them. This ordering can also be guessed, contributing a factor of at most $(k-1)!$ for each agent, which is absorbed into the $|V|^{O(k^2)}$ term.
    
    With a fixed guess (including the ordering of crossings for each agent), the problem is reduced to finding a set of non-crossing directed paths. For each agent $i$, we need to connect their terminal $s_i$ to the entry of their first crossing path, connect the exit of one crossing path to the entry of the next, and connect the exit of the last crossing path to $t_i$. This is an instance of \textsc{Steiner Orientation}. The set of terminal pairs for this subproblem, $\mathcal{T}'$, has size $O(k^2)$.
    
    If the \textsc{Steiner Orientation} instance returns a solution, we construct the full path $P_i$ for each agent. 
    The \textsc{Steiner Orientation} subproblem can be solved in $|V|^{O(|\mathcal{T'}|)} = |V|^{O(k^2)}$ time~\cite{CyganKN13}. The total runtime is the number of guesses multiplied by the time for verification: $|V|^{O(k^2)} \cdot |V|^{O(k^2)} = |V|^{O(k^2)}$. 
    We finally return the valid guess of the minimum cost.
\end{proof}

\subsection{Parameterized by number of edges}

\begin{theorem}\label{thm:FPT:E}
    A minimum-cost solution for a \ccMAR{} instance $(G,\terminals)$ can be computed in time $2^{2^{O(|E|)}}(k|V|)^{O(1)}$.
\end{theorem}

\begin{proof}
    We reduce an instance $(G, \mathcal{T})$ to an equivalent instance $(G', \mathcal{T}')$ whose size is bounded by a function of $|E|$, and this reduction is performed in FPT time by $|E|$. 

    We construct the reduced graph $G'=(V', E', A')$. Let $V_E \subseteq V$ be the set of endpoints of the undirected edges in $E$ and let $r =|V_E|$ be the number of vertices incident to undirected edges. Note that $r = |V_E| \le 2|E|$. We then prepare two vertex sets $S' = \{s'_{U}\mid U\subseteq V_E\}$ and $T = \{t'_{U}\mid U\subseteq V_E\}$, which are associated with the subsets of $V_E$.The vertex set of $G'$ is defined as $V' = V_E\cup S'\cup T'$. 
    The edge set $E'$ of $G'$ is the same as $E$, i.e., $E'=E$.
    We also define the arc set $A'_E$ such that $(u,v)\in A'_E$ where $u,v\in V_E$ if and only if a directed path from $u$ to $v$ exists in the directed subgraph $G-E$ obtained by removing all edges from $G$.
    Finally, we connect $S'$ with $V_E$ and $T'$ with $V_E$, respectively. 
    For each vertex $s_U\in S'$, we add arcs to vertices in $U\subseteq V_E$. Similarly, we add arcs from vertices in $U\subseteq V_E$ to $t_U\in T'$. 
    We denote by $A'_{S'}$ and $A'_{T'}$ the set of arcs from $S'$ to $V_E$ and $V_E$ to $T'$, respectively, and define $A' = A'_E\cup A'_{S'}\cup A'_{T'}$.
    The constructed graph $G'$ has at most $r + 2\cdot 2^r$ vertices, at most $r^2$ edges, and at most $2^{r+1}r + r^2$ arcs.
    Since the reachability check can be computed in $O(|G|)$, the graph $G'$ can be constructed in time $2^{O(r)}|G|$.
    
    The new terminal pairs $\mathcal{T}'$ are defined on $G'$. For each original agent $i$, the new start terminal is $s_U\in S'$ where $U$ is the set of vertices in $V_E$ such that $s_i\in V$ can reach in $G-E$. Similarly, the new end terminal is $t_U\in T'$ where $U$ is the set of vertices in $V_E$ that can reach $t_i\in V$ in $G-E$.
    Since defining new terminal pairs is done in time $2^{O(r)}k|V|^{O(1)}$, the total running time for constructing the new instance is $2^{O(r)}k|V|^{O(1)}$.
    
    New instance $(G', \mathcal{T}')$ is equivalent to the original. A path in $G$ corresponds to a path in $G'$ by contracting its directed segments into single arcs, and vice versa. 
    Since crossing costs only occur on the undirected edges $E$, which are common to both graphs, the cost of any solution is preserved. 

    Since the number of vertices in $G'$ is at most $2^{O(r)}$, by Proposition~\ref{prop:FPT:by:graph:size}, 
    a minimum-cost solution can be found in time $2^{2^{O(r)}}k^{O(1)}$. 
    Therefore, the total running time is $2^{2^{O(|E|)}}(k|V|)^{O(1)}$.
    \end{proof}


\subsection{Parameterized by number of arcs plus $\boldsymbol{k}$}
\begin{theorem}
   A minimum-cost solution for a \ccMAR{} instance $(G,\terminals)$ can be computed in time $|A|^{O(k|A|)}|V|^{O(1)}$.
\end{theorem}

\begin{proof}
    Consider any path $P$ in $G$. We will say that $P$ has the \textit{path form} $a_1,\ldots, a_{m}$ if $\{a_i\mid i \in [m]\}$ is the set of arcs appearing in $P$ and, for any $0<i<j\le m$, $a_i$ appear before $a_j$, while moving through $P$.

    Note that, we have at most $2^{|A|}|A|!$ possible path forms. 
    Additionally, given a terminal pair $(s_i,t_i)$ and a path form $a_1,\ldots,a_m$, we can check in polynomial time if there exists an $s_i-t_i$ path that has the given form. Indeed, this can be done by verifying that the following three conditions hold:
    \begin{enumerate}
        \item Either $a_1=(s_i,v)$, for some $v\in V$ or $a_1=(u,v)$ and there exists an $s_i-u$ path $G-A$.
        \item Either $a_{m}= (u,t_i)$ for some $u\in V$ or $a_{m}=(u,v)$ and there exists an $v-t_i$ path $G-A$.
        \item For all $i \in [m-1]$, either $a_{i}= (u,v)$ and $a_{i+1}= (v,w)$ for some $u,w\in V$ or $a_{i}= (u,v)$, $a_{i+1}= (v',w)$ and there exists an $v-v'$ path $G-A$.        
    \end{enumerate}
    \noindent     
    
    We proceed by guessing the path form of each $s_i-t_i$ path in an optimal solution. This give us $2^{k|A|}(|A|!)^k$ possible guesses.
    Since it is easy to check the validity of our guesses, we can assume that we keep only the valid ones.

    We proceed by defining a sub-problem called Pre-Described \ccMAR. 
\begin{description}
    \item[Input:] An instance $(G,\terminals)$ of \ccMAR. Additionally, for each $(s_i,t_i)\in\terminals$ path form $a^i_1,\ldots,a^i_{m_i}$ and an integer $\ell$.
    \item[Question:] Are there $k$ paths $P_i$ such that, for all $i \in k$, $P_i$ is an $s_i$-$t_i$ path with path form $a^i_1,\ldots,a^i_{m_i}$ and $\cost(\{P_1,\ldots,P_k\})$ is at most $\ell$.
\end{description}

    \begin{claim}
        Pre-Described \ccMAR can be solved in polynomial time.
    \end{claim}

    \begin{proofclaim}
    We just need to find the minimum cost way to complete the path forms using only undirected edges. To do so we create the following instance:

    For each pair $(s_i,t_i)$ let $a^i_1,\ldots,a^i_{m_i}$ be the path form of the $s_i-t_i$ path we search.
    We create a set of terminal pairs that will replace $(s_i,t_i)$ as follows:
    \begin{itemize}
        \item if $a^i_1 = (u,v)$ and $u \neq s_i$ we create the pair $(s_i,u)$
        \item if $a^i_{m_i} = (u,v)$ and $v \neq t_i$ we create the pair $(v,t_i)$
        \item if $a^i_j=(u,v)$, $a^i_{j+1}=(v',w)$ and $v\neq v'$ we create the pair $(v,v')$.
    \end{itemize}

    We repeat this process for all pairs in $\terminals$. Let $\terminals'$ be the new set of terminals. Notice that, in order to create $s_i$-$t_i$ paths that respect the given path forms we need to connect all the terminal pairs in $\terminals'$ with undirected paths. 

    We can compute the minimum cost paths between all pairs in $\terminals'$ by solving the instance $(G-A,\terminals')$. This can be done in polynomial time as any pair of vertices has a unique path in $G-A$ since $G-A$ is a forest.
    
    Using the solution of $(G-A,\terminals')$ with the path forms we have guessed we have the solution that agrees with the path forms and has minimum possible crossing cost. 
    \end{proofclaim}

    Note that, each of our guesses defines an instance of Pre-Described \ccMAR\ problem.
    Since we can compute a solution of each such an instance polynomial time, we can compute an optimal solution for \ccMAR by solving all instances and keeping the solution of minimum cost. 
    
    The only part of our algorithm that is not polynomial is the guesses of the path forms. Therefore, the algorithm terminates in $2^{k|A|}(|A|!)^k |V|^{O(1)} = |A|^{O(k|A|)}|V|^{O(1)}$. 
\end{proof}
\subsection{Parameterized by number of arcs plus diameter}

\begin{theorem}
   A minimum-cost solution for a \ccMAR{} instance $(G,\terminals)$ can be computed in FPT time when parametrized by $|A|+diam(G)$, where $diam(G)$ is the diameter of $G$.
\end{theorem}


\begin{proof}
    We first prove that a mixed graph bounded $A+diam(G)$ has a bounded number of vertices.
    Recall that by Observation~\ref{obs:input:simpl} we can assume that the given graph is acyclic and has no vertex of degree $1$.

    Note that $G-A$ is a forest with at most $2|A|$ leaves. Also, $\Delta(G-A)\le 2|A|$ as any tree has maximum degree at most equal to the number of leaves. Therefore, the maximum degree in $G$ is at most $3|A|$. Finally, any graph of diameter $diam(G)$ and maximum degree $\Delta$ has $1+diam(G)\Delta^{diam(G)}$ vertices~\cite{DiestelGT0030488}.

    Finally, since the number of vertices in $G$ are bounded by a function of $|A|+diam(G)$, using Proposition~\ref{prop:FPT:by:graph:size} results to an algorithm that runs in FPT time, parameterized by $|A|+diam(G)$. This completes the proof.
\end{proof}

\section{Minimum-Cost Solution on Restricted Settings}
\subsection{Parameterization by vertex cover number in unweighted graphs}
In this section, we consider unweighted graphs. We first show that the problem becomes tractable when parameterized by the vertex cover number plus $k$. Then, we show that the vertex cover number alone suffices if we additionally assume that terminal pairs do not overlap.

\begin{theorem}\label{thm:vc:plus:k}
    \ccMAR{} on unweighted graphs is FPT when parameterized by the vertex cover number $\vc$ plus the number of agents $k$.
\end{theorem}

\begin{proof}
The main idea is to construct a small equivalent instance whose size is exponential in the parameter $\vc + k$. We begin by computing a vertex cover $S$ of size at most $2\vc$ (achievable in polynomial time~\cite{Bar-YehudaE81}). We then augment $S$ to include all terminal vertices, so its new size is bounded by $|S| \le 2\vc + 2k$.

By the definition of a vertex cover, $I = V\setminus S$ is an independent set. This means any path that enters a vertex in $I$ from a vertex in $S$ must immediately return to a vertex in $S$.

We can categorize vertices in $I$ by their "type," which is defined by their neighborhood in $S$. Specifically, two vertices $v_1, v_2 \in I$ are of the same type if they are connected to the exact same subset of vertices in $S$ via undirected edges, incoming arcs, and outgoing arcs. For any vertex $v \in I$, its neighborhood is defined by three subsets of $S$. Therefore, the total number of distinct types is at most $2^{|S|} \cdot 2^{|S|} \cdot 2^{|S|} = 8^{|S|}$.

Now, consider a set of vertices in $I$ that are all of the same type. Since there are only $k$ agents, at most $k$ of these functionally identical vertices can be used in any given solution. Any additional vertices of this type are redundant. Therefore, we can safely reduce the number of vertices of each type to just $k$.

This reduction rule creates a small instance with at most $(2\vc + 2k) + 8^{|S|} \cdot k$ vertices.
By Proposition~\ref{prop:FPT:by:graph:size}, the problem is FPT when parameterized by $\vc+k$.
\end{proof}

\begin{comment}
\begin{theorem}\label{thm:poly:kernel:vc:plus:k}
    \ccMAR{} admits a kernel of size $O(\vc^4+k)$. \tesshu{Is this unweighted case?}
\end{theorem}

\textit{Sketch Proof.} While the complete proof is rather technical, the kernelization algorithm's main idea is straightforward. We begin by computing a vertex cover $S$ of size at most $2\vc$. The algorithm relies on two crucial observations. First, any path connecting vertices $u,v \in S$ that avoids other vertex cover vertices must either move directly from $u$ to $v$ or pass through exactly one intermediate vertex $w$ from $V\setminus S$. Second, for any pair $(u,v) \in S\times S$, if there exists a path containing a subpath $u$-$w$-$v$ where $w \notin S$ is a non-terminal vertex, then $w$ can serve as a universal intermediate vertex for all $u$-$v$ paths. Since we do not know which vertices are used as intermediate vertices and there are $O(\vc^2)$ vertex pairs in $S\times S$, it suffices to keep $O(\vc^2)$ such intermediate vertices per pair. The selection of these vertices can be made almost arbitrarily. Combining these selected vertices with the terminal vertices yields our kernel.
\end{comment}


\begin{theorem}
    Let $(G,\terminals)$ be an instance of \ccMAR{} where $G$ is an unweighted graph and no vertex of $V(G)$ appears in more than one terminal pair. A minimum-cost solution of $(G,\terminals)$ can be computed in FPT time when parameterized by the vertex cover number $\vc$ of $G$. 
\end{theorem}

\begin{proof}
    Recall that we can assume that $G$ is acyclic due to Observation~\ref{obs:input:simpl}.
    Let $G$ be a graph with vertex cover number $\vc$.
    Let $S \subseteq V$ be a vertex cover of $G$ of size $\vc$. It can be computed in time $2^{\vc}n^{O(1)}$ \cite{0001N24}. Also, let $I = V \setminus S$ be the corresponding independent set.
    
    
    We now define types of terminal pairs. In particular:
    \begin{itemize}
    \item Any pair $(s,t) \in \terminals$ such that ${s,t} \cap S \neq \emptyset$ has a unique type.
    \item Two terminal pairs $(s,t), (s',t') \in \terminals$, where $s,t,s',t' \in I$, have the same type if:
    \begin{itemize}
    \item $N^+(s) = N^+(s')$, $N^-(s) = N^-(s')$, and $N(s) = N(s')$, and
    \item $N^+(t) = N^+(t')$, $N^-(t) = N^-(t')$, and $N(t) = N(t')$.
    \end{itemize}
    \end{itemize}
    Note that this gives us at most $\vc$ unique types, and at most $16^\vc$ additional types. Indeed, the additional types are at most $16^\vc$, as the number of different neighborhoods for vertices in $I$ is $4^\vc$.
    
    Using the types of pairs, we define types of vertices as follows:
    \begin{itemize}
    \item For each type of terminal pair, we define two vertex types: one for the source vertices of this type and one for the sink vertices. This gives us at most $16^{\vc+1} + 2\vc$ types.
    \item Two non-terminal vertices $u,v \in I$ have the same type if $N^+(u) = N^+(v)$, $N^-(u) = N^-(v)$, and $N(u) = N(v)$. This gives us at most $4^\vc$ additional types.
    \item Finally, we assign a unique type to each non-terminal vertex in $S$, giving at most $\vc$ additional types.
    \end{itemize}
    
    Note that each vertex type includes either non-terminal vertices with the same neighborhood or terminal vertices that are all either sources or sinks of a given terminal pair type. Also, each vertex in $S$ has its own unique type.
    
    Finally, if a vertex type contains more than one vertex, then all vertices of that type are incident to only one edge. Indeed, otherwise we could create a cycle, which contradicts the assumption that $G$ is acyclic.
    
    We proceed by defining types of paths in $G$ using the previously defined types of vertices. To do so, we first note that any path has length at most $2\vc + 1$, and no two consecutive vertices belong to $I$. This holds because the vertices in $I$ form an independent set.
    
    This gives us at most $(2^\vc \vc!)(16^{\vc+1} + 4^\vc + 1)^{\vc+1}$ types of paths, where the first factor accounts for the vertices from $S$ that appear in the path and their ordering, while the second factor accounts for whether a vertex from $I$ appears between vertices of $S$, and if so, what type it has.
    
    Having defined types of terminals, vertices, and paths, we are now ready to prove that there exists an optimal solution satisfying some very restricted properties.
    
    Before stating the claim, we define what we refer to as the \textit{intermediate vertices} of a path. Given a $u$–$v$ path $P$, we define the intermediate vertices of $P$ as the set $V(P) \setminus {u, v}$.
    \begin{claim}\label{claim:bounded:number:same:type:vertices:as:intermediate}
        There exist an optimal solution $\mathcal{P}$ of $(G,\terminals)$ where, from any given type of vertices in $I$ at most two vertices of this type appear as intermediate vertices in the paths in  $\mathcal{P}$.
    \end{claim}

    \begin{proofclaim}
        We start with an optimal solution $\mathcal{P}$. Let $\nu$ be a type of vertices in $I$ such that at least three vertices of type $\nu$ appear in $\mathcal{P}$ as intermediate vertices. We will show how to modify the paths so that only two vertices of type $\nu$ appear as intermediate vertices, without changing any of the vertices of any other type.
        
        Note that each vertex of type $\nu$ is incident to at most one edge. Indeed, otherwise we could create a cycle, which contradicts the assumption that $G$ is acyclic.
        
        We consider two cases: either each vertex of type $\nu$ has zero edges incident to it, or each has exactly one edge incident to it.
        
        \textbf{Case 1} (each vertex is incident to zero edges):
        Fix a vertex $u$ of the given type. Consider any appearance of a vertex of type $\nu$ as an intermediate vertex; let $v$ be such a vertex. It is safe to replace $v$ with $u$. Indeed, since $u$ and $v$ have the same type, they have the same neighborhood. Also, the cost does not change, since we only replace arcs in the paths. We can repeat this until only $v$ appears as an intermediate vertex of type $\nu$.
        
        \textbf{Case 2} (each vertex is incident to one edge):
        Let $w$ be the common neighbor of all vertices of type $\nu$, such that ${v, w} \in E$ for all vertices $v$ of type $\nu$.
        
        We define two sets of vertices, $V_\nu^+$ and $V_\nu^-$. We set $V_\nu^+$ to be the set of vertices of type $\nu$ such that:
        \begin{itemize}
        \item $v$ appears as an intermediate vertex in a path $P \in \mathcal{P}$, and
        \item $v$ appears exactly before $w$ in $P$.
        \end{itemize}
        We define $V_\nu^-$ to be the set of vertices of type $\nu$ such that:
        \begin{itemize}
        \item $v$ appears as an intermediate vertex in a path $P \in \mathcal{P}$, and
        \item $v$ appears exactly after $w$ in $P$.
        \end{itemize}
        
        We first modify the paths so that $V_\nu^+$ contains at most one vertex. If $|V_\nu^+| \leq 1$, we do not need to do anything. Therefore, assume that $|V_\nu^+| \geq 2$, and select arbitrarily a vertex $v \in V_\nu^+$. For each path in $\mathcal{P}$, we proceed with the following modification:
        \begin{itemize}
        \item For each $P \in \mathcal{P}$,
        \item Find any vertex of type $\nu$ that appears before $w$ in $P$,
        \item Replace this vertex with $v$.
        \end{itemize}
        
        Let $\mathcal{P}'$ be the set of paths after these modifications.
        Note that, since vertices of the same type have the same neighborhood, $\mathcal{P}'$ is a feasible solution for $(G, \terminals)$.
        We now show that it is also an optimal solution.
        
        We show that starting from an optimal solution, replacing just one vertex results in another optimal solution.
        Let $u$ be the vertex we replaced with $v$ in a path $P \in \mathcal{P}$. Also, let $P'$ be a path in which $v$ appears before $w$ in $\mathcal{P}$ (such a path exists since $v \in V_\nu^+$). The only change in the set of edges appearing in $\mathcal{P}$ is that we replaced ${u, w}$ with ${v, w}$.
        
        Let $x_1$ be the number of appearances of the edge ${w, u}$ (oriented as $(w, u)$) in the paths of $\mathcal{P}$, and let $x_2$ be the number of appearances of the edge ${w, v}$ (oriented as $(w, v)$). Note that $x_1 = x_2$. Indeed, otherwise, replacing $u$ with $v$ in $P$, or replacing $v$ with $u$ in $P'$, would result in a feasible solution with lower cost.
        
        Therefore, replacing $u$ with $v$ in $P$ results in a new optimal solution. This means that the previous process creates a sequence of optimal solutions until $V_\nu^+ = {v}$.
        
        Using symmetric arguments, we can reduce $V_\nu^-$ to a singleton.
        
        This shows that the only vertices of type $\nu$ appearing as intermediate vertices are those in the sets $V_\nu^+$ and $V_\nu^-$.
        
        Since the previous process affects only vertices of the same type, we can repeat it for all types that appear more than twice as intermediate vertices. Thus, the claim holds.
    \end{proofclaim}
    
    Hereafter, we assume that we have an optimal solution that satisfies the properties of Claim~\ref{claim:bounded:number:same:type:vertices:as:intermediate}.
    Since each terminal pair is related to exactly $2$ types of vertices, for each type of terminals, at most $4$ pairs of each type have their vertices as intermediate vertices in the considered optimal solution.
    Now, we will show that we can predetermine which terminal pairs are allowed to be used as intermediate vertices in paths.
    Before that, we will define the sets of vertices that we would like to allow as intermediate vertices.
    
    For each type of terminals $\tau$, let $\terminals_\tau \subseteq \terminals$ be the set of all terminal pairs of type $\tau$. We select (arbitrarily) $\min {4,|\terminals_\tau|}$ pairs of terminals of type $\tau$; let $\terminals^*_\tau$ be the set of these terminals.

    \begin{claim} \label{claim:select:terminals:as:intermediate}
        There exists an optimal solution $\mathcal{P}$ where any terminal vertex $v$ that appears as an intermediate vertex in a path in $\mathcal{P}$ belongs to a set $\terminals^*_{\tau}$.
    \end{claim}

    \begin{proofclaim}
        Note that we need to prove the claim for types $\tau$ such that $|\terminals_{\tau}| > 4$. Consider a type $\tau$ and
        a pair $(s,t)\in \terminals_\tau \setminus \terminals^*_\tau$ such that either $s$ or $t$ appear as intermediate vertices in a path in $\mathcal{P}$.
        We will show how to modify $\mathcal{P}$ in order to replace all appearances of $s$ and $t$ as intermediate vertices with appearances of $s'$ and $t'$ where $(s',t')\in \terminals^*_\tau$. Also, this process will not affect any other vertices.

        Note that, since in $\mathcal{P}$ at most $4$ pairs of type $\tau$ are used as intermediate vertices and $(s,t)\in \terminals_\tau \setminus \terminals^*_\tau$ there exists at least one pair $(s',t')\in \terminals^*_\tau$ where neither $s'$ nor $t'$ appear as intermediate vertices in any path in $\mathcal{P}$. 
        
        Since $(s,t)$ and $(s',t')$ are of the same type of terminals, by replacing all appearances of $s$ with $s'$ and $t$ with $t'$ we obtain a feasible solution of the same exact crossing cost. To see this, observe that any additional cost may appear because of edges incident to $s$, $t$, $s'$, and $t'$. 
        However since we have replace all appearances, then any cost related to some edge $\{v,w\}$, for $v\in \{s,t\}$ and a $w\in V$ was already appearing in $\mathcal{P}$ because of the edge $\{v',w\}$, for some $v'\in \{s',t'\}$. Similarly, any cost related to some edge $\{v',w\}$, for $v'\in \{s',t'\}$ and a $w\in V$ was already appearing in $\mathcal{P}$ because of the edge $\{v,w\}$, for some $v\in \{s,t\}$.

        Finally, we repeat this until the only intermediate vertices that belong in terminals of type $\tau$ are these in $\terminals^*_\tau$. Also, this does not affect any other terminal types. Therefore, we can repeat this for all terminal types.        
    \end{proofclaim}
    
    Now, we prove a similar claim about the appearances of vertices from $I$ that are not terminal vertices.
    Let, for each type $\nu$ of vertices that does not include a terminal vertex, $V_\nu$ be the set of all vertices of type $\nu$.
    For each such type, we select (arbitrarily) $\min {2,|V_\nu|}$ vertices of type $\nu$; let $V^*_\nu$ be the set of selected vertices.
        
    \begin{claim}\label{claim:select:non:terminals:as:intermediate}
        There exists an optimal solution $\mathcal{P}$ where any non-terminal vertex $v$ that appears in a path in $\mathcal{P}$ belongs in a set $V^*_{\nu}$.
    \end{claim}

    \begin{proofclaim}        
        The proof follows the same arguments as the proof of Claim~\ref{claim:select:terminals:as:intermediate}
    
        We need to prove the claim for types $\nu$ such that $|V_\nu| > 2$. Consider a type $\nu$ and
        a vertex $v \in V_\nu \setminus V^*_\nu$ that appears in a path in $\mathcal{P}$.
        Since we consider non-terminal vertices, we know that these vertices appear only as intermediate vertices in $\mathcal{P}$.
        We will show how to modify $\mathcal{P}$ in order to replace all appearances of $v$ with appearances of $v'$, where $v' \in V^*_\nu$. Also, this process will not affect any other vertices.
        
        By the claim~\ref{claim:bounded:number:same:type:vertices:as:intermediate}, since in $\mathcal{P}$ at most $2$ vertices of type $\nu$ are used as intermediate vertices and $v \in V_\nu \setminus V^*_\nu$, there exists at least one vertex $v' \in V^*_\nu$ that does not appear as an intermediate vertex in any path in $\mathcal{P}$.
        
        Since $v$ and $v'$ are of the same type of vertices, by replacing all appearances of $v$ with $v'$, we obtain a feasible solution of the exact same crossing cost. To see this, observe that any additional cost is due to edges incident to $v$ and $v'$.
        However, since we have replaced all appearances, any cost related to some edge ${v,w}$, for some $w \in V$, was already present in $\mathcal{P}$ due to the edge ${v',w}$. Similarly, any cost related to some edge ${v',w}$, for some $w \in V$, was already present in $\mathcal{P}$ because of the edge ${v,w}$. Therefore, the case is the same, and the new solution is also optimal.
        
        We repeat this until the only intermediate vertices of type $\nu$ are those in $V^*_\nu$. Also, this does not affect any other types. Therefore, we can repeat this for all types.
    \end{proofclaim}

    The Claims~\ref{claim:select:terminals:as:intermediate} and~\ref{claim:select:non:terminals:as:intermediate} allow us to avoid guessing which vertices of each type are used as intermediate vertices in an optimal solution, as we can select (arbitrarily) the vertices of each type that will be used as intermediate vertices.
    
    The final step before we proceed with the algorithm is to show that, for each type, the terminal pairs that do not appear as intermediate vertices can use the same type of paths.

    \begin{claim}\label{claim:identical:paths:at:most:5:per:type}
        There exists an optimal solution $\mathcal{P}$ such that, for each type of terminals $\tau$ and pair $(s,t),(s',t')\in \terminals_\tau\setminus\terminals_\tau^*$ the $s$-$t$ path and $s'$-$t'$ path in $\mathcal{P}$ have the same path type.
    \end{claim}

    \begin{proofclaim}
        Consider an optimal solution that satisfies the claims~\ref{claim:select:non:terminals:as:intermediate} and~\ref{claim:select:terminals:as:intermediate}. Consider a type of terminals $\tau$ and $(s,t),(s',t')\in \terminals_\tau\setminus\terminals_\tau^*$ such that the $s$-$t$ path $P$ and the $s'$-$t'$ path $P'$ in $\mathcal{P}$ do not have the same path type.
        We will prove that we can modify only one of $P$, $P'$ so that they both have the same type. Let $Q$ and $Q'$ be the paths we obtain from $P$ and $P'$, respectively, by swapping the vertices $s$ and $s'$, and $t$ and $t'$. Note that $Q$ is an $s'$-$t'$ path that has the same type as $P$, and $Q'$ is an $s$-$t$ path that has the same type as $P'$.
        
        We claim that both $(\mathcal{P}\setminus{P})\cup{Q'}$ and $(\mathcal{P}\setminus{P'})\cup{Q}$ are optimal solutions.
        To see this, let $e_1,\ldots,e_p$ be the edges appearing in $P$ that also appear, in the opposite direction, in other paths of $\mathcal{P}$. Also, let $x_i$, $i \in [p]$, be the number of times the edge $e_i$ appears in a path of $\mathcal{P}$ in the opposite direction to the one it appears in $P$.
        Similarly, let $e'_1,\ldots,e'_q$ be the edges appearing in $P'$ that also appear, in the opposite direction, in other paths of $\mathcal{P}$, and let $x'_j$, $j \in [q]$, be the number of times the edge $e'_j$ appears in a path of $\mathcal{P}$ in the opposite direction to the one it appears in $P'$.
        
        Note that none of the edges $e_i$, $i \in [p]$, or $e'_j$, $j \in [q]$, are incident to $s$, $s'$, $t$, or $t'$. This is because none of these vertices are used as intermediate vertices in any path, and since we have no overlapping terminals, they appear only in their own paths.
        
        By construction of $Q$ and $Q'$, the following two inequalities hold:
        \begin{align*}
        \cost(\mathcal{P})\le \cost((\mathcal{P}\setminus\{P\})\cup\{Q'\}) \le \cost(\mathcal{P})-\sum_{i\in [p]} x_i +\sum_{j \in [q]}x'_j, 
        \end{align*}
        \begin{align*}\cost(\mathcal{P})\le \cost((\mathcal{P}\setminus\{P'\})\cup\{Q\}) \le \cost(\mathcal{P})-\sum_{j\in {q}} x'_j +\sum_{i \in [p]}x_i. 
        \end{align*}
        Therefore, we can conclude that $\sum_{j \in [q]} x'_j = \sum_{i \in [p]} x_i$ and all three solutions are optimal.
        By selecting one of the pairs $(s,t)\in \terminals_\tau\setminus \terminals_\tau^*$, we can repeat the previous process until all other pairs use paths of the same path type as $(s,t)$. This does not affect the path types used by terminals of other terminal types, so we can apply it independently for each terminal type. Thus, the claim holds.
    \end{proofclaim}

    Now, we are ready to provide the algorithm.
    First, we fix the sets $V^*_\nu$ and $\terminals_\tau^*$, for all types of non-terminal vertices and terminal pairs, respectively.
    We consider a terminal type $\tau$. Due to Claim~\ref{claim:identical:paths:at:most:5:per:type} we know that there exists a solution where at most $5$ different types of paths connect terminal pairs of type $\tau$: one for each pair in $\terminals^*\tau$ and one for the rest. Note that, 
    we can guess the types that connect all pairs in FPT time, as the number of types of paths is bounded by a function of $\vc$. 
    
    After guessing the paths, we need to decide which are the intermediate vertices. Any path has at most $\vc+1$ vertices from $I$, and we know the type of all such vertices. Therefore, for each guess, we have at most $5\vc+5$ intermediate vertices that are unknown. For the non-terminal intermediate vertices, we have at most $2$ possible options. Therefore, we can guess them, needing at most $2^{5\vc+5}$ guesses. For the terminal vertices, after we check the type, we need to check which terminal type includes this type of vertex. Since we have at most $4$ terminal pairs for each type and the types of vertices depend on whether the vertex is a source or sink, we only have $4$ possible guesses. Therefore, at most $4^{5\vc+5}$ guesses.
    After these guesses, we have fixed all the paths that connect any pair of type $\tau$.
    
    Since we have $16^{\vc}$ many different types of terminals, we can guess these paths for all different types.
    After the guesses, we have created FPT many feasible solutions. Also, since we have made all possible guesses, we have created at least one optimal solution. Since we can compute the cost of a feasible solution in polynomial time, it suffices to compute the cost of all the feasible solutions we created and keep the smallest.
\end{proof}

\subsection{Overwhelming number of agents}

In this section, we show that, if almost all pairs of vertices appear as terminal pairs, then we can decide if there exists a $0$-cost solution. 

\begin{theorem}
We can decide whether a \ccMAR{} instance $(G,\terminals)$ admits a $0$-cost solution in FPT time parameterized by $|\mathcal{N}|$, where $\mathcal{N} = (V \times V) \setminus \terminals$.
\end{theorem}

\begin{proof}
Consider an edge $e = \{u, v\} \in E$. We claim that either the cost of any solution is at least $1$, or at least one of $(u, v)$, $(v, u)$ belongs to $\mathcal{N}$.

Assume that neither $(u, v)$ nor $(v, u)$ belongs to $\mathcal{N}$. Then both $(u, v), (v, u) \in \terminals$. Since we can assume that $G$ is acyclic, there exists neither a $u$-$v$ path excluding $\{u, v\}$ nor a $v$-$u$ path excluding $\{u, v\}$. Therefore, both paths must use $\{u, v\}$, and no solution of cost $0$ exists. Note that this condition can be verified in polynomial time.

Hereafter, we assume that for any $e = \{u, v\} \in E$, at least one of $(u, v)$ or $(v, u)$ belongs to $\mathcal{N}$, thus $|\mathcal{N}| \geq |E|$. By Theorem~\ref{thm:FPT:E}, we can compute the optimal solution in FPT time parameterized by $|E| \leq |\mathcal{N}|$. 
\end{proof}

\section{Conclusion}
We introduced the Crossing Cost Multi-Agent Routing (\ccMAR) model, a reasonable framework for asynchronous multi-agent routing where head-on interactions are captured via a soft cost function rather than hard constraints. This setting reflects real-world scenarios where agents operate without global synchronization and where congestion arises from opposing edge traversals.
Our analysis combines game-theoretic and algorithmic perspectives. We showed that Nash equilibrial always exist and are reachable via convergent dynamics, with polynomial-time computability under mild assumptions, and PLS-completeness in the general case. On the optimization side, we first observe that our problem generalizes the Steiner Orientation problem and thus is NP-hard. To overcome this, we developed parameterized algorithms with FPT and XP guarantees for various structural parameters, such as the number of arcs, edges, and terminal pairs, as well as vertex cover number.

This framework opens several directions for future research. In particular, while we provide FPT algorithms for parameters involving the number of arcs and edges, the parameterized complexity with respect to $|A|$ alone remains unresolved. Further exploration of efficient heuristics or approximation schemes for practical instances remains an open and practically relevant challenge.
\bibliographystyle{plain}
\bibliography{ref}
\end{document}